\documentclass{article}

\PassOptionsToPackage{numbers, compress}{natbib}
\usepackage[preprint]{neurips_2023}

\usepackage{times}  %
\usepackage{helvet}  %
\usepackage{courier}  %
\usepackage[hyphens]{url}  %
\usepackage{graphicx} %
\usepackage{caption} %
\usepackage{todonotes}

\usepackage{algorithm}
\usepackage{xcolor}
\usepackage{algpseudocode}
\usepackage{amsmath,amsthm,amssymb}
\usepackage{newfloat}
\usepackage{listings}
\DeclareCaptionStyle{ruled}{labelfont=normalfont,labelsep=colon,strut=off} %
\floatstyle{ruled}
\newfloat{listing}{tb}{lst}{}
\floatname{listing}{Listing}
\newtheorem{theorem}{Theorem}[section]
\newtheorem{corollary}{Corollary}[theorem]
\newtheorem{lemma}[theorem]{Lemma}
\newtheorem{definition}[theorem]{Definition}

\algnotext{EndFor}
\algnotext{EndIf}

\newcommand{\conpcheck}{\ensuremath{\mathsf{CoNPCheck}}}
\newcommand{\twoqbfcheck}{\ensuremath{\mathsf{2QBFCheck}}}
\newcommand{\threeqbfcheck}{\ensuremath{\mathsf{3QBFCheck}}}
\newcommand{\Vars}{\ensuremath{\mathsf{Vars }}}
\newcommand{\Cell}[2]{\ensuremath{\mathsf{Cell}_{\langle #1, #2 \rangle}}}

\newcommand{\Cnt}[2]{\ensuremath{\mathsf{Cnt}_{\langle #1, #2 \rangle}}}
\newcommand{\satisfying}[1]{\ensuremath{sol({#1})}} %

\newcommand{\prob}{\ensuremath{\mathsf{Pr}}}
\newcommand{\expect}{\ensuremath{\mathsf{E}}}

\newcommand{\bigO}{\ensuremath{\mathcal{O}}}

\newcommand{\sol}[1]{|\satisfying{#1}|}   %

\newcommand{\val}{v}   %
\newcommand{\Cest}{\mathsf{CntEst}} %
\newcommand{\hashstock}{\mathsf{hashassgn}_{stock}}
\newcommand{\hashcells}{\mathsf{hashassgn}_{cells}}

\newcommand{\chigh}{c_\mathsf{high}}
\newcommand{\clow}{c_\mathsf{low}}

\newcommand{\NP}{\mathsf{NP}}
\newcommand{\coNP}{\mathsf{CoNP}}

\newcommand{\myP}{\mathsf{P}}
\newcommand{\Sigmatwop}{{\Sigma_2^\myP}}

\newtheorem{proposition}[theorem]{Proposition}
\newtheorem{claim}[theorem]{Claim}

\newcommand{\vecbold}[1]{\ensuremath{{\boldsymbol{#1}}}}

\newcommand{\toolname}{\ensuremath{\mathsf{AFCounter}}}

\newcommand{\fullversion}[1]{#1}
\newcommand{\shortversion}[1]{}

\title{Auditable Algorithms for Approximate Model Counting
\thanks{
	The authors decided to forgo the old convention of alphabetical ordering of authors in favor of a
	randomized ordering, denoted by \textcircled{r}. The publicly verifiable record of the randomization is available at
	\protect\url{https://www.aeaweb.org/journals/policies/random-author-order/search}
}
}
 \author{
Kuldeep S. Meel \textsuperscript{\rm 1} \textcircled{r}
Supratik Chakraborty \textsuperscript{\rm 2} \textcircled{r}
S. Akshay\textsuperscript{\rm 2} \\ \\ 
    \textsuperscript{\rm 1} University of Toronto, Canada\\
     \textsuperscript{\rm 2} Indian Institute of Technology Bombay, Mumbai, India 
 }

\begin{document}

\maketitle

\begin{abstract}
Model counting, or counting the satisfying assignments of a Boolean
formula, is a fundamental problem with diverse applications.  Given
$\#P$-hardness of the problem, developing algorithms for approximate
counting is an important research area. Building on the practical
success of SAT-solvers, the focus has recently shifted from theory to
practical implementations of approximate counting algorithms. This has
brought to focus new challenges, such as the design of auditable
approximate counters that not only provide an approximation of the
model count, but also a certificate that a verifier with limited
computational power can use to check if the count is indeed within the
promised bounds of approximation.

Towards generating certificates, we start by examining the best-known
deterministic approximate counting algorithm~\cite{stockmeyer} that
uses polynomially many calls to a $\Sigma_2^P$ oracle. We show that
this can be audited via a $\Sigma_2^P$ oracle with the query constructed over 
$n^2 \log^2 n$ variables, where the original formula has $n$
variables.  Since $n$ is often large, we ask if the count of variables
in the certificate can be reduced -- a crucial question for potential
implementation. We show that this is indeed possible with a tradeoff
in the counting algorithm's complexity. Specifically, we develop new
deterministic approximate counting algorithms that invoke a
$\Sigma_3^P$ oracle, but can be certified using a $\Sigma_2^P$ oracle
using certificates on far fewer variables: our final algorithm uses only $n \log n$ variables.
Our study demonstrates that one can simplify auditing significantly if
we allow the counting algorithm to access a slightly more powerful
oracle. This shows for the first time how audit complexity can be
traded for complexity of approximate counting.
\end{abstract}
\section{Introduction}
\label{sec:intro}
Trust in computation done by third-party software is increasingly
becoming a concern in today's world.  As an example, consider a
software implementing an algorithm for which we already have a
rigorous proof of correctness.  Even when the software is implemented
by trusted and expert programmers, one cannot be sure that the
implementation will never produce an erroneous result.  %
Insidious bugs may go undetected for several reasons: for
example, the implementation may be too complex to have been tested
comprehensively, or it may use code from untrusted libraries that have
bugs, or the programmer may have inadvertently introduced a bug.
In such a scenario, we wish to ask: \emph{How can we provide assurance
that an untrusted implementation of an algorithm has produced a
correct result for a given problem instance?}

One way to answer the above question is to formally model the semantics of
every instruction in the implementation in a theorem proving system
and then prove that the composition of these instructions always
produces the correct result.
This is a challenging and skill-intensive exercise for anything but
the simplest of software, but once completed, it gives very
high assurance about the correctness of results produced by the
implementation \emph{for all instances of the problem being solved.}
A relatively less daunting alternative is to have the algorithm (and
its implementation) output for each problem instance, not only the
result but also a \emph{certificate} of correctness of the result.
The certificate must be independently checkable by an auditor, and a
successful audit must yield a proof of correctness of the computed
result \emph{for the given problem instance}.  We call such
algorithms \emph{auditable algorithms}. Note that a successful audit
for one problem instance does not prove bug-freeness of the
implementation; it only shows that the result computed for this
specific problem instance is correct.  A key requirement for effective
use of auditable algorithms is access to an independent trusted
certificate auditor. Fortunately, a certificate auditor is usually
much simpler to design, and its implementation far easier to prove
correct than the implementation of the original algorithm.  Therefore,
auditable algorithms provide a pragmatic solution to the problem of
assuring correctness for untrusted implementations of algorithms.  Not
surprisingly, auditable algorithms have been studied earlier in the
context of different
problems~\cite{mehlhorn-auditable,Heule13,Wetzler14,Gocht22,Ekici17,Barthe06,Necula98,Magron15,Cheung16}.

In this paper, we present the first detailed study of auditable
algorithms for model counting with provable approximation guarantees--
an important problem with diverse applications.  Our study reveals a
nuanced landscape where the complexity of an auditable algorithm can
be ``traded off'' for suitable measures of complexity of certificate
auditing.  This opens up the possibility of treating certificate audit
complexity as a first-class concern in designing auditable
algorithms.

Before we delve into further details, a brief background on model
counting is useful.  Formally, given a system of constraints
$\varphi$, model counting requires us to count the solutions (or
models) of $\varphi$.  This problem is known to be computationally
hard~\cite{Valiant79}, and consequently, 
significant effort has therefore been invested over the
years in designing approximate model counters.
Over the last
decade, riding on spectacular advances in propositional satisfiability
(SAT) solving techniques (see~\cite{sat-technique} for details),
approximate model counting technology has matured to the point where
we can count solutions of propositional constraints with several
thousands of variables within a couple of hours on a standard laptop,
while providing strong approximation guarantees~\cite{approxmc}.  This
has led to a surge in the number of tools from different application
domains that have started using approximate model counters as
black-box
engines~\cite{klebanov2014,teuber2021quantifying,beck2020automating,zinkus2023mcfil}
to solve various problems.  With increasing such usage comes
increasing demands of assurance on the correctness of counts computed
by approximate counters.  Auditable approximate counting offers a
promising approach to solve this problem.

Approximate model counting algorithms come in two different flavours,
viz. deterministic and randomized.  While randomized counting
algorithms with probably approximately correct (PAC)
guarantees~\cite{Valiant84,Valiant84CACM} have been found to scale
better in practice, { it is not possible to provide a meaningful
certificate of correctness for the result computed by a single run of
a randomized algorithm (see~\cite{mehlhorn-auditable} for a nice
exposition on this topic)}.  Therefore, we focus on deterministic
algorithms for approximate model counting.  Specifically, we ask: (a)
what would serve as a meaningful certificate for the best-known
deterministic approximate counting algorithm~\cite{stockmeyer}?
(b) what is the fine-grained complexity of auditing such a
certificate?  (c) can we re-design the approximate counting algorithm
so as to make certificate auditing more efficient?  and (d) is there a
tradeoff in the complexity of an auditable deterministic approximate
counting algorithm and the fine-grained complexity of certificate
auditing?  Our study shows that there is a nuanced interplay between
the complexity of auditable approximate counting algorithms and that
of auditing the certificates generated by such algorithms. Therefore,
if efficiency of certificate auditing is treated as a first-class
concern, we must re-think the design of auditable
algorithms for approximate counting.  Indeed, we present two new
auditable approximate counting algorithms that
significantly improve the efficiency of certificate auditing at the
cost of making the counting algorithms themselves a bit less
efficient.

The primary contributions of the paper can be summarized as listed
below.  In the following, we wish to count models of a propositional
formula on $n$ variables, upto a constant approximation factor.  %
\begin{enumerate}
\item We first show that the best-known deterministic approximate
counting algorithm (from~\cite{stockmeyer}) can be converted to an
auditable algorithm that generates a certificate over $\bigO(n^2\log^2
n)$ variables.  The resulting counting algorithm uses
$\bigO(n\log n)$ invocations of a $\Sigma_2^P$ oracle, as in
Stockmeyer's original algorithm.  An auditor, however,
needs only an invocation of a $\Sigma_2^P$ oracle on formulas
constructed over $\bigO(n^2\log^2 n)$ variables.
\item Next, we develop an deterministic approximate counting algorithm for which the certificate uses only $\bigO(n^2)$ variables.  This reduction is achieved by allowing the counting algorithm access to a more powerful $\Sigma_3^P$ (instead of a
$\Sigma_2^P$) oracle, that is invoked only $\bigO(n)$ times.
Certificate auditing in this case requires a single invocation of a
$\Sigma_2^P$ oracle on a formula constructed over $\bigO(n^2)$
variables.
\item Finally, we present another deterministic approximate
counting algorithm that requires a certificate over only $\bigO(n \log
n)$ variables.  This makes certificate auditing significantly more
efficient compared to the cases above, when viewed through the lens of
fine-grained complexity.  This improvement is obtained by allowing the
counting algorithm to invoke a $\Sigma_3^P$ oracle $\bigO(n \log n)$
times. Certificate auditing now requires one invocation of a
$\Sigma_2^P$ oracle on formulas over $\bigO(n \log n)$ variables.
\end{enumerate}
Our results show that the fine-grained complexity of certificate
checking can be ``traded off'' significantly against the power of
oracles accessible to a determistic approximate counting algorithm.
This opens up the possibility of designing new counting algorithms
that are cognizant of certificate checking complexity. Our
proofs involve reasoning about specially constructed formulas with
qantifier alternations, with incomplete information about parts of the
formula.  
This is quite challenging in general, and one of our technical novelties lies in our application of the probabilistic method in the proofs.

The remainder of the paper is organized as follows.  We present
some preliminaries and notation in Section~\ref{sec:prelim}, and
formalize the notion of an audit for approximate counting
algorithms.  In Section~\ref{sec:stock}, we discuss how Stockmeyer's algorithm can be turned into an auditable algorithm, and also present the corresponding audit algorithm. Sections~\ref{sec:bgp} and \ref{sec:combined} discuss the
design of two new auditable algorithms for approximate counting
with successively small certificate sizes, and the complexity
of auditing these certificates.  Finally, we conclude with
a discussion of our findings in Section~\ref{sec:conclusion}. %

\section{Preliminaries}
\label{sec:prelim}
Let $F$ be a Boolean formula in conjunctive normal form (CNF), and let
$\Vars(F)$, also called its \emph{support}, be the set of variables
appearing in $F$.  An assignment $\sigma$ of truth values to the
variables in $\Vars(F)$ is called a \emph{satisfying assignment}
or \emph{witness} of $F$ if it makes $F$ evaluate to true.  We denote
the set of all witnesses of $F$ by $\satisfying{F}$ and its count by
$\sol{F}$. Throughout the paper, we use $n$ to denote $|\Vars(F)|$.
We also use bold-faced letters to represent vectors of Boolean
variables, when there is no confusion.

We write $\prob\left[\mathcal{Z}: {\Omega} \right]$ to denote the
probability of outcome $\mathcal{Z}$ when sampling from a probability
space ${\Omega}$.  For brevity, we omit ${\Omega}$ when it is clear
from the context.  The expected value of $\mathcal{Z}$ is denoted
$\expect\left[\mathcal{Z}\right]$ and its variance is denoted
$\sigma^2\left[\mathcal{Z}\right]$. The quantity
$\frac{ \sigma^2\left[\mathcal{Z}\right]}{\expect\left[\mathcal{Z}\right]}$
is called the dispersion index of the random variable $\mathcal{Z}$.
Given a distribution $\mathcal{D}$, we use
$\mathcal{Z} \sim \mathcal{D}$ to denote that $\mathcal{Z}$ is sampled
from the distribution $\mathcal{D}$.

We introduce some notation from complexity theory for later use. Let
$\mathcal{C}$ be a decision complexity class.  An \emph{oracle} for
$\mathcal{C}$ is a (possibly hypothetical) computational device that
when queried with any problem in $\mathcal{C}$, answers it correctly
in one time unit.  If $\mathcal{C}_1$ and $\mathcal{C}_2$ are two
complexity classes, we use $\mathcal{C}_1^{\mathcal{C}_2}$ to denote
the class of problems solvable by an algorithm for a problem in
$\mathcal{C}_1$ with access to an oracle for $\mathcal{C}_2$.  The
complexity classes $\myP$, $\NP$ and $\coNP$ are well-known. The
classes $\Sigma_2^P$ and $\Sigma_3^P$ are defined as $\NP^{\NP}$ and
$\NP^{\Sigma_2^P}$ respectively.

\paragraph{\bfseries Deterministic Approximate Counting (DAC):}
Given a Boolean formula $F$ in CNF and an approximation factor $k$ ($> 1$), \emph{deterministic approximate
counting} (or DAC for short), requires us to compute an estimate
$\Cest~(\geq 0)$ using a deterministic algorithm such that
$\frac{\sol{F}}{k} \leq \Cest \leq \sol{F} \cdot k$.
It has been shown in~\cite{stockmeyer} that DAC is in
$\myP^{\Sigmatwop}$, the crucial idea in the proof being use of
pairwise independent hash functions to partition the space of all
assignments.  Given a DAC algorithm for any $k > 1$, we can apply it
on independent copies of $F$ conjoined together to obtain $\sol{F}$ to
within any approximation factor $> 1$. Therefore, we will often be
concerned with DAC algorithms for some convenient $k$, viz. $8$ or
$16$.

Practical implementations of DAC algorithms must of course use SAT/QBF
solvers as proxies for oracles like the $\Sigmatwop$ oracle
in~\cite{stockmeyer}. Therefore, it is important in practice to look
at the finer structure of oracle invocations, especially the support
sizes and exact quantifier prefix of the formulas fed to the oracles.

\paragraph{\bfseries $k$-wise independent hash family:}

For $n,m\in \mathbb{N}^{\ge 0}$, let $\mathcal{H}(n,m)$ denote a family
of hash functions from $\{0,1\}^n$ to $\{0,1\}^m$. We use
$h \xleftarrow{R} \mathcal{H}(n,m)$ to denote the probability space
obtained by choosing $h$ uniformly at random from
$\mathcal{H}(n,m)$.

\begin{definition}
	For $k \ge 2$, we say that $\mathcal{H}(n,m)$ is a family of
        $k$-wise independent hash functions if
        for every distinct $\vecbold{y_1}, \ldots \vecbold{y_k} \in \{0,1\}^n$ and for every (possibly repeated) $\vecbold{\alpha_1}, \ldots \vecbold{\alpha_k} \in \{0,1\}^m$,
$		\Pr [ h(\boldsymbol{y_1}) = \boldsymbol{\alpha_1} \wedge  \cdots h(\boldsymbol{y_k}) = \boldsymbol{\alpha_k}] = \left(\frac{1}{2^m}\right)^k,
$
where the probability space is $h \xleftarrow{R} \mathcal{H}(n,m)$.
Such a family is also sometimes denoted $\mathcal{H}(n,m,k)$ to highlight the degree of independence.%
\end{definition}	

Significantly, every $h \in \mathcal{H}(n,m,k)$
can be specified by a $k$-tuple of coefficients ($a_1, a_2, \ldots
a_k$) from the finite field $GF(2^{max(n,m)})$ and vice
versa~\cite{BGP2000}.  Hence, $\bigO\big(k.\max(m,n)\big)$ bits
suffice to represent $h$.  We leverage this observation to simplify
notation and write $\exists h$ as syntactic sugar for existentially
quantifying over the $\bigO\big(k.\max(m,n)\big)$ Boolean variables
that represent $h$.

For every formula $F$ on $n$ variables, and for every
$\vecbold{\alpha} \in \{0,1\}^m$, define $\Cell{F}{h,\vecbold{\alpha}}$ to be the
(sub-)set of solutions of $F$ that are mapped to $\vecbold{\alpha}$ by $h$.
Formally, $\Cell{F}{h,\vecbold{\alpha}} = \{\vecbold{\sigma} \mid \vecbold{\sigma} \in \{0,1\}^n$,
$\vecbold{\sigma}\models F$ and $h(\vecbold{\sigma}) = \vecbold{\alpha}\}$.  For 
convenience, we use $\Cnt{F}{h,\vecbold{\alpha}}$ to denote the cardinality of
$\Cell{F}{h,\vecbold{\alpha}}$.

The following result proves useful in our analysis later.
\begin{proposition}{\cite{BR1994}}
	Let $t \geq 4$ be an even integer. Let $Z_1, Z_2, \ldots Z_k$
	be $t$-wise independent random variables taking values in
	$[0,1]$. Let $Z = Z_1 + Z_2 + \ldots Z_k$, $\mu
	= \expect[Z]$. Then
$        \Pr[|Z-\mu| \geq \varepsilon \mu ] \leq
                    8 \cdot \left(\frac{t\mu +t^2}{\varepsilon^2 \mu^2}\right)^{t/2}
$
for all $\varepsilon > 0$.
\end{proposition}

\subsection{Audit Complexity}

Let $A$ be an auditable algorithm that solves (or claims to solve) DAC
with a known constant-factor approximation, say $k$.  Let
$\mathcal{K}$ denote the certificate generated by $A$ for an input
formula $F$ over $n$ variables.  We expect $\mathcal{K}$ to be a
string of Boolean values of size that is a function of $n$.  We can
now define the notion of an auditor for $A$ as follows.

\begin{definition} An auditor for $A$ is an algorithm $B$ that
  takes as inputs the formula $F$,
  the estimate $\Cest$ returned by $A$, and the
  certificate $\mathcal{K}$.  Algorithm $B$ uses an
  oracle of complexity class $\mathcal{C}$ to certify that $\Cest$
  lies within the claimed $k$-factor approximation of $\sol{F}$.  The {\em
    audit complexity of algorithm $A$} is defined to be the tuple ($\mathcal{C}, r, t$), where the auditor can certify the answer
  computed by $A$ by making $t$ calls to an oracle of complexity class
  $\mathcal{C}$, where each call/query is over a formula with at most
  $r$ variables.
\end{definition}
In our context, motivated by Stockmeyer's result~\cite{stockmeyer}, we restrict
$\mathcal{C}$ to be $\Sigmatwop$, and further fix the number of calls $t$ to be $1$. Hence, we omit $\mathcal{C}$ and $t$, and simply refer to $r$ as the audit complexity of the algorithm $A$.

\section{Stockmeyer's Algorithm and its Audit}
\label{sec:stock}
In this section, we consider a celebrated DAC algorithm due to Stockmeyer~\cite{stockmeyer}, and analyze its audit complexity. Towards this end, we
first consider a QBF formula that plays a significant role in Stockmeyer's algorithm.   For a given formula Boolean formula $F$ on $n$ variables, and for $1 \le m \le n$, consider
  \begin{align*}
 	\varphi_{stock}^F(m):=  \exists h_1\ldots \exists h_m, \forall \vecbold{\vecbold{z_1}}, \forall \vecbold{\vecbold{z_2}}~~~~~~~~~~~~~~~~~~~~~~~~~~~~~~~~~~~~~\\ \bigvee_{i=1}^m \big( F(\vecbold{\vecbold{z_1}})\land (h_i(\vecbold{\vecbold{z_1}})=h_i(\vecbold{\vecbold{z_2}}))\rightarrow \neg F(\vecbold{\vecbold{z_2}})\big)
 \end{align*}
 
In the above formula, the $h_i$'s are hash functions from
$\mathcal{H}(n, m, 2)$.  The formula says that there exist $m$ such hash
functions such that for any two assignments
$\vecbold{z_1}, \vecbold{z_2}\in \{0.1\}^n$, if $\vecbold{z_1}$ is a
solution for $F$ and some hash function $h_i$ maps $\vecbold{z_1}$ and
$\vecbold{z_2}$ to the same cell, then $\vecbold{z_2}$ cannot be a
solution. In other words, for each solution some hash function puts it
in a cell where it is the only solution. {Recall that $\exists h_i$ is
shorthand for existentially quantifying over coefficients of $h$.}

Now, the following claim and lemma provide necessary and sufficient conditions for $\varphi_{stock}^F(m)$ evaluating to True. 

\begin{claim}[Necessary]
 	\label{cl:stockup}
 	If $\varphi^F_{stock}(m)=1$, then  $\sol{F}\leq m\cdot 2^m$.
\end{claim}
\fullversion{
\begin{proof} To see this, note that if $\varphi_{stock}^F(m)=1$, there
  	exists an injective map from $\sol{F}$ to $[m] \times \{0,1\}^m$.
  	The solution $z$ of $F$ can be uniquely identified a pair $(i,h_i(z))$ where $h_i$ is the hash function that puts it in a bucket $h_i(z)$. Hence, we get $\sol{F}\leq m\cdot 2^m$.
  \end{proof}
 }

The following lemma capturing sufficient conditions is inherent in Sipser's work~\cite{Sipser83}; to obtain the precise constant for our analysis, we provide an alternate simpler argument \shortversion{(proof in Appendix)}.
\begin{lemma}[Sufficient]
	\label{lm:stockexist}
If $\sol{F}\leq 2^{m-2}$, then $\varphi_{stock}^F(m) =1$ 
\end{lemma}

\fullversion{
\begin{proof}
For %
$h_i \xleftarrow{R} H(n,m,2)$ and a fixed $y \in \satisfying{F}$, let us define indicator variable $\gamma_{z,y,h_i}$ for all $z \in \satisfying{F} \setminus \{y\}$ as
\begin{align*}
	\gamma_{z,y,h_i} = 
	\begin{cases}
		1 & h_i(z) = h_i(y)\\
		0 & \text{otherwise}
	\end{cases}
\end{align*} 

Since $h_i$ is chosen from pairwise independent 
hash family, we have $\Pr[\gamma_{z,y,h_i} = 1] \leq \frac{1}{2^m}$.  Let $\Gamma_{F,y,h_i} = \sum_{z\in \satisfying{F}\setminus \{y\}} \gamma_{z,y,h_i}$. We have $\expect[\Gamma_{F,y,h_i}] = \sum\limits_{z	\in \satisfying{F} \setminus \{y\}} \expect[\gamma_{z,y,h_i}] = \frac{\sol{F}-1}{2^m} \leq \frac{1}{4}$, since $\sol{F}\leq 2^{m-2}$.
Now, by Markov's inequality, we have 
\begin{align*}
	\Pr[\Gamma_{F,y,h_i} \geq 1 ] \leq \expect[\Gamma_{F,y,h_i}] \leq \frac{1}{4}
\end{align*}

Therefore, we have 
\begin{align*}
	\Pr\left[ \bigwedge_{i=1}^m  \{\Gamma_{F,y,h_i} \geq 1\} \right] \leq \frac{1}{4^m}	
\end{align*} 

Now, taking union bound, we have 
\begin{align*}
	\Pr\left[ \exists y \in \satisfying{F},  \bigwedge_{i=1}^m  \{\Gamma_{F,y,h_i} \geq 1\} \right] \leq \frac{2^{m-2}}{4^m} \leq \frac{1}{2^m}	
\end{align*} 
Taking the complement of the above statement, we have 
\begin{align*}
	\Pr\left[ \forall y \in \satisfying{F},  \bigvee_{i=1}^m  \{\Gamma_{F,y,h_i} < 1\} \right] \geq 1- \frac{1}{2^m} > 0	
\end{align*}

Therefore, we have $\varphi_{stock}^F(m) =1$.
\end{proof}
}
We can now write Stockmeyer's algorithm as follows:

\begin{algorithm}
	\caption{Stock$(F)$}
		\label{algo:stock}
	\begin{algorithmic}[1]
		\State $\val \gets 0$; $\mathsf{hashassgn}_{stock}\gets \emptyset,\Cest\gets 0$;
		\State $F' \gets \mathsf{MakeCopies}(F,\log n)$ \hfill / {$F'$ has $n\log n $ variables}
		\For{m=1 to $n\log n$}
		\State  $(\mathsf{ret},\mathsf{assign})\gets \twoqbfcheck (\varphi_{stock}^{F'}(m)$);\label{line:qbf-call}
		\If{$\mathsf{ret}==1$
		}\label{line:equalcells-check}
		\State $\val=m$;
		\State $\mathsf{hashassgn}_{stock}=\mathsf{assign}$;
		\EndIf
		\State \textbf{break}
		\EndFor
		\State $\Cest=2^{(\frac{\val}{\log (n)})}$;\\
		\Return $(\val, \Cest,\hashstock)$ 
	\end{algorithmic}

\end{algorithm}

Equipped with necessary and sufficient conditions for the satisfiability of $\sol{F}\leq 2^{m-2}$, we now describe Stockmeyer's algorithm, as presented in Algorithm~\ref{algo:stock}.
We first invoke  
$\mathsf{MakeCopies}(F,\log n)$ subroutine, which makes $\log(n)$ many copies of $F$, each with a fresh set of variables and conjuncts them.
Now, the algorithm runs over $m$ from $1$ to number of variables of $F'$, i.e., $n\log n$ and checks if $\varphi_{stock}^{F'}(m)$ evaluates to 1. Each such call is a \twoqbfcheck\  as $\varphi_{stock}^{F'}$ is a $\Sigma_2^P$ formula.

Now, if $\varphi_{stock}^{F'}(i)=1$, then $\varphi_{stock}^{F'}(i+1)=1$ as well, so the above algorithm essentially searches for the first point (smallest value of $m$) where $\varphi_{stock}^{F'}(i)=1$ and returns the estimate of number of solutions at that point as $2^{\frac{\val}{\log n}}$. In other words, when algorithm breaks at line 8, we have $\varphi_{stock}^{F'}(\val)=1$ and $\varphi_{stock}^{F'}(\val-1)=0$.

\begin{theorem}[Stockmeyer]
	Algorithm~\ref{algo:stock} makes O($n \log n$) queries to $\Sigma_2^P$ oracle and solves DAC. 
\end{theorem}

\fullversion{
\begin{proof}
	
From Claim~\ref{cl:stockup}, we have $\sol{F'}\leq \val\cdot 2^\val$. Taking the contrapositive of Lemma~\ref{lm:stockexist} at $\val-1$, we get $\sol{F'}>2^{(\val-1)-2}$. Combining the two, we have $2^{\val-3}\leq  \sol{F'}\leq \val \cdot 2^\val$, i.e., an $8v$-factor approximation.

Now, we can see that this results in a constant-factor approximation for $\sol{F}$, due to having taken $\log(n)$ copies of $F$. First note that $\sol{F'}=(\sol{F})^{\log(n)}$. Second we note that $\val\leq n \log n\leq n^2$ and so $\frac{1}{n^2}\leq \frac{1}{v}$. Now, by above argument we have,

\begin{align*}
  2^{\val-3}\leq \sol{F'}\leq \val\cdot 2^\val\\
  \frac{\sol{F'}}{n^2}\leq  \frac{\sol{F'}}{\val}\leq 2^\val\leq \sol{F'}\cdot 2^{3}\\
       {\left(\frac{\sol{F}}{16}\right)}^{\log(n)}\leq (2^{\frac{\val}{\log(n)}})^{\log(n)}\leq {(\sol{F})}^{\log(n)} \cdot 2^{3}\\
       \frac{\sol{F}}{16} \leq \Cest\leq 16\cdot \sol{F} \ \text{for $n\geq 4$}
\end{align*}

Thus, Algorithm~\ref{algo:stock}  computes a 16-factor approximation of $\sol{F}$ for $n \ge 4$, and makes $\bigO\big(n \log n\big)$ calls to a $\Sigma_2^P$-oracle.
\end{proof}
}
\subsection{Auditing Stockmeyer's algorithm}

Now to check the result/count given by the above algorithm, what can we do? In other words, what could be considered as a {\em certificate} for the answer? If we inspect the above algorithm, we can see that the crux is to identify the smallest value of $m$ for which 2QBFCheck returns True. 

Thus, at $\val$, we need to check that the True answer is correct, i.e., the $\val^{th}$ call to the $\Sigma_2^P$ oracle. This is easy since we can use the certificate returned by the call, i.e., the set of hash functions $h_1, \ldots, h_\val$ that was used, i.e., $\hashstock$. Given this certificate (indeed a solver returns this) we can pass it to the verifier and hence, a query to NP oracle suffices. 

However, this does not suffice. We also need to check that at $\val-1$, the False answer is correct, i.e., at $\val-1$, we need to check that $\varphi_{stock}^{F'}(\val-1)=0$, which can be answered by a query to $\Sigma_2^P$ oracle. More precisely, for $m=\val-1$, we need to check that $\neg \varphi^{F'}_{stock}(\val-1)=1$, i.e., $\forall h_1\ldots \forall h_m$, $\exists \vecbold{z_1}, \exists \vecbold{z_2} \bigwedge_{i=1}^m (F(\vecbold{z_1})\land (h_i(\vecbold{z_1})=h_i(\vecbold{z_2}))\land F(\vecbold{z_2}))$.

\begin{algorithm}
	\caption{StockAudit$(F,\mathsf{Stock}(F))$}
	\begin{algorithmic}[1]
\State {$\mathsf{poscheck}\gets 0$; $\mathsf{negcheck}\gets 0$;}
\State  $\mathsf{poscheck}\gets \conpcheck (\varphi_{stock}^{F'}(\val)[h_1\ldots,h_v \gets \hashstock])$; \label{line:conpcheck}
                \State  $\mathsf{negcheck}\gets \twoqbfcheck(\neg\varphi_{stock}^{F'}(\val-1))$;\label{line:twoqbfcheck}   
		\If{$\mathsf{poscheck}\land \mathsf{negcheck}==1$}\label{line:equalcheck}                            
                \State \Return $\mathsf{Verified}$;
\EndIf
\end{algorithmic}
        \label{algo:stock-audit}
\end{algorithm}

Thus, the complexity is dominated by the ``False'' check at $\val-1$. In fact combining the two checks as done in line 4, we just require a single call to 2-QBF solver. Thus the \emph{audit-complexity}, or the number of variables on which this call is made is $n^2\log^2(n)+ 2 n\log n$. %
Thus we obtain:

\begin{corollary}
\label{cor:stockaudit}
The audit complexity of Algorithm~\ref{algo:stock} is $O(n^2 \log^2 n)$.
\end{corollary}

\fullversion{
\begin{proof}
  Algorithm~\ref{algo:stock-audit} gives the audit algorithm for
  Algorithm~\ref{algo:stock}.  The two checks in lines $2$ and $3$ of
  Algorithm~\ref{algo:stock-audit} can be combined into one
  check as follows.

  Let $\varphi_{st\_aud}^{F'}(v)$ be defined as
  \begin{align*}
  \forall \boldsymbol{z_1} \forall \boldsymbol{z_2} \bigvee_{i=1}^v \big(F'(\boldsymbol{z_1}) \wedge (h_i(\boldsymbol{z_1}) = h_i(\boldsymbol{z_2})) \rightarrow \neg F'(\boldsymbol{z_2})\big)[h_1, \ldots h_v \gets \hashstock]\\
 \bigwedge \forall h_1 \ldots \forall h_{v-1} \exists \boldsymbol{z_1} \exists \boldsymbol{z_2} \bigwedge_{i=1}^{v-1} \big(F'(\boldsymbol{z_1}) \wedge (h_i(\boldsymbol{z_1}) = h_i(\boldsymbol{z_2})) \wedge F'(\boldsymbol{z_2})\big)~~~~~~~~~~~~~~~~~~~~~~~~~~~~
  \end{align*}
\normalsize
Recalling that $\forall h_i$ is simply syntactic sugar for quantifying
over coefficients of the hash functions, the above formula is a formula
with quantifier prefix $\forall^* \exists^*$ and its truth can be
decided by one invocation of a $\Sigma_2^P$ oracle.

To determine the audit complexity, notice that each $\boldsymbol{z_i}$
in the above formula represents $n \log n$ bits, since $F'$ has a
support of size $n \log n$.  Similarly, each $h_i$ is a hash function
from $\mathcal{H}(n\log n, v-1, 2)$, where $v \le n \log n$. Since
$v \le n \log n$, by our discussion in Section~\ref{sec:prelim}, each
such hash function can be represented using $\bigO(n \log n)$ bits.
Therefore, the total number of Boolean variables in the quantifier
prefix of $\varphi_{staud}^{F'}(v)$ is in $\bigO\big( (n \log n)^2 +
n\log n\big) = \bigO(n^2 \log^2 n)$.
\end{proof}
}

\section{Trading Algorithm vs Audit Complexity}
\label{sec:bgp}
Observe that the dominant factor of $n^2\log n$ in the audit complexity of Stockmeyer's algorithm was the need to certify ``False'' returned by 2QBFCheck. Therefore, a plausible approach to improving the audit complexity would be to develop an algorithm whose audit does not have to rely on certifying any of the ``False'' values returned by the underlying QBF checker.  To this end, we construct a new formula $\varphi$ whose validity would imply both an lower and upper bound (without requires a separate check for False as in the previous case). At a high-level, our approach is to partition the solution space into cells and then assert that every cell has neither too many solutions nor too few solutions. 

\subsubsection*{Not Too Many Solutions}

For a given input Boolean formula $F$ over $n$ variables, consider the formula:
\begin{align*}
  &\exists h \in \mathcal{H}(n,m,n) \, \forall \vecbold{\alpha} \in \{0,1\}^m, \, \vecbold{y_1}, \vecbold{y_2}, \ldots \vecbold{y_{u+1}}\in \{0,1\}^n\\
&\mathsf{ConstructNotMany}(F,h, m, \vecbold{\alpha}, \vecbold{y_1},\vecbold{y_2}, \ldots \vecbold{y_{u+1}})\\
&\text{\it where, }
\mathsf{ConstructNotMany}(F,h, m, \vecbold{\alpha}, \vecbold{y_1},\vecbold{y_2}, \ldots \vecbold{y_{u+1}}):=\\
&  \left(\bigwedge_{i=1}^{u+1} (F(\vecbold{y_i}) \wedge h(\vecbold{y_i}) = \vecbold{\alpha}) \rightarrow  \bigvee_{i=1}^{u} (\vecbold{y_{u+1}} = \vecbold{y_i}) \right)\nonumber
\end{align*}

The above formula encodes that there is a $n$-wise indepedent hash function such that for each $m$-sized cell $\vecbold{\alpha}$, and every set of $u+1$ solutions (each solution being an assignment to $n$ variables), if all of them are mapped to the same cell, then two of them must be identical. So there must be at most $u$ solutions in each cell. 

\subsubsection*{At Least Few Solutions}
Next, let us consider the formula
\begin{align*}
&\exists h \in \mathcal{H}(n,m,n) \, \forall \vecbold{\alpha} \in \{0,1\}^{m} \, \exists \vecbold{z_1}, \vecbold{z_2}, \ldots \vecbold{z_{\ell}}\in\{0,1\}^n\\
&\mathsf{ConstructAtLeastFew}(F, h, m, \vecbold{\alpha},\vecbold{z_1},\vecbold{z_2}, \ldots \vecbold{z_{\ell}})\\
&\text{\it where, }
  \mathsf{ConstructAtLeastFew}(F, h, m, \vecbold{\alpha},\vecbold{z_1},\vecbold{z_2}, \ldots \vecbold{z_{\ell}}):=\\
&  \left(\bigwedge_{i=1}^{\ell} (F(\vecbold{z_i}) \wedge h(\vecbold{z_i}) = \vecbold{\alpha} \wedge \bigwedge_{i=1,j=i+1}^{\ell} ({z_i} \neq \vecbold{z_j})
  \right) 
  	\nonumber
\end{align*}

This formula says that there is a hash function such that for each cell $\vecbold{\alpha}$, there are $\ell$ solutions which all map to the cell $\vecbold{\alpha}$. Hence, if they are all distinct, in each cell the number of solutions is at least $\ell$. Thus the number of solutions is bounded below by $\ell$ times number of cells.

Combining these, for any Boolean formula $F$, parameters $\ell, u$, we obtain a single formula asserting existence of a hash function such that each cell has a number of solutions that is simultaneously bounded from above by $u$ and below by $\ell$.%
\begin{align*}
 & \varphi^{\langle F, \ell, u \rangle}_{Cells}(m):=\exists h \forall \vecbold{\alpha}, \vecbold{y_1}, \vecbold{y_2}, \ldots \vecbold{y_{u+1}} \exists \vecbold{z_1}, \vecbold{z_2}, \ldots \vecbold{z_{\ell}} \\ 
  &\big( \mathsf{ConstructNotMany}(F,h, m, \vecbold{\alpha}, \vecbold{y_1}, \vecbold{y_2}, \ldots \vecbold{y_{u+1}})\nonumber  \\
  & \wedge \mathsf{ConstructAtLeastFew}(F, h, m, \vecbold{\alpha}, \vecbold{z_1}, \vecbold{z_2}, \ldots \vecbold{z_{\ell}})\big)
\end{align*}
Note that %
the formula has a $\exists\forall\exists$ alternation, i.e., it is a $\Sigma_3^P$ formula. Then,
\begin{proposition}
\label{prop:bgp}
If $\varphi^{\langle F,\ell,u\rangle}_{Cells}(m)=1$, then $\ell\cdot 2^m \leq \sol{F}\leq u\cdot 2^m$.
\end{proposition}
\fullversion{
\begin{proof}
The proof follows from the definition of $\varphi^{\langle
F,\ell,u\rangle}_{Cells}(m)$. Specifically,
\begin{align*}
\exists h \in \mathcal{H}(n,m,n) \, \forall \vecbold{\alpha} \in \{0,1\}^m, \, \vecbold{y_1}, \vecbold{y_2}, \ldots \vecbold{y_{u+1}}\in \{0,1\}^n\\
\mathsf{ConstructNotMany}(F,h, m, \vecbold{\alpha}, \vecbold{y_1},\vecbold{y_2}, \ldots \vecbold{y_{u+1}})
\end{align*}
is true implies that $\sol{F} \le u.2^m$. Similarly, 
\begin{align*}
\exists h \in \mathcal{H}(n,m,n) \, \forall \vecbold{\alpha} \in \{0,1\}^{m} \, \exists \vecbold{z_1}, \vecbold{z_2}, \ldots \vecbold{z_{\ell}}\in\{0,1\}^n\\
\mathsf{ConstructAtLeastFew}(F, h, m, \vecbold{\alpha},\vecbold{z_1},\vecbold{z_2}, \ldots \vecbold{z_{\ell}})
\end{align*}
is true implies $\sol{F} \ge l.2^m$.
\end{proof}
}
Assuming that it is possible to show that such a hash function $h$ with polynomial time evaluation and polynomial size description indeed exists (we will do so in the proof below), we now design an algorithm for approximate model counting. The algorithm simply goes over every value of $m$ from 1 to $n$ and checks whether  $ \varphi^{\langle F, \ell, u\rangle}_{Cells}(m)$ is True. We present the pseudocode of the algorithm in Algorithm~\ref{algo:smallcells}. 
\begin{algorithm}[h]
	\caption{$\mathsf{EqualCellsCounter}(F)$}
	\label{algo:smallcells}
	\begin{algorithmic}[1]
		\State $u \gets 16384n$;
		\State $\ell \gets 1024n$;
                \State $\hashcells\gets 0;\Cest\gets 0$;
		\For{$m=1$ to $n$}
               	\State $\mathsf{(ret,assign)} \gets \threeqbfcheck(\varphi^{\langle F,\ell,u\rangle}_{Cells}(m))$; 
		\If{$\mathsf{ret}$==1}\label{line:eqcells-check}
                \State $\hashcells=\mathsf{assign}$;
		\EndIf
                \State \textbf{break}
		\EndFor
                \State $\Cest=\ell\cdot 2^m$;\\
                \Return $(\Cest,\hashcells)$ %
	\end{algorithmic}
\end{algorithm}

\begin{theorem}
\label{thm:intermediate}
Algorithm~\ref{algo:smallcells} solves DAC and makes $O(n)$ calls to a $\Sigma_3^P$ oracle. 
\end{theorem}
\fullversion{
\begin{proof}
First observe that by definition of $\varphi^{\langle F,\ell, u\rangle}_{Cells}(m)$, and Proposition~\ref{prop:bgp}, when the check in line~\ref{line:eqcells-check} in Algorithm~\ref{algo:smallcells} passes, we must have
$\ell \cdot 2^m \leq \sol{F} \leq u \cdot 2^m$. 
Since $\frac{u}{\ell} = 16$, we obtain a 16-factor approximation, i.e.,
$\frac{\sol{F}}{16}\leq \Cest \leq \sol{F}$.
Further, observe that we make at most $n$ calls to  ${\Sigma_3^P}$ oracle.

It remains to show that there indeed exists $h$ such that with the
choice of $\ell$ and $u$ as above, for some $m \in [n]$, the check in
line~\ref{line:eqcells-check} will pass. To this end, we will utilize
the probabilistic method and show that for some $m \in \{1, \ldots
n\}$, the probability that there exists a hash function
$h \in \mathcal{H}(n,m,n)$ that satisfies $\varphi^{\langle F,\ell,
u\rangle}_{Cells}(m)$ is $> 1$.  .

Let us fix a cell $\vecbold{\alpha}$. Let $m= \lfloor \log \sol{F} -
12 - \log n \rfloor$. For $\vecbold{y} \in \{0,1\}^n$.  Now, define
the indicator variable $\gamma_{\vecbold{y},\vecbold{\alpha}}$ such
that
$\gamma_{\vecbold{y},\vecbold{\alpha}} = \begin{cases}
	1 & h(\vecbold{y}) = \vecbold{\alpha} \\
	0 & \text{otherwise}
\end{cases}$. 
Therefore, $\mu_{m}
= \sum\limits_{\vecbold{y} \in \satisfying{F}} \expect[\gamma_{\vecbold{y},\vecbold{\alpha}}]
= \frac{|\satisfying{F}|}{2^m}$. Recalling that
$m= \lfloor \log \sol{F} - 12 - \log n \rfloor$, we have $\mu_{m} \in
[4096n,8192n] $. Now consider $h\xleftarrow{R}\mathcal{H}(n,m,n)$
i.e., chosen randomly from a $n$-wise independent hash family. Then
$\{\gamma_{\vecbold{y},\vecbold{\alpha}}\}$ are $n$-wise independent.
Assuming $n \ge 4$, we have
\begin{align*}
	&\Pr[\Cnt{F}{h,\vecbold{\alpha}} < 1024n \text{ or } \Cnt{F}{h,\vecbold{\alpha}} > 16384n]\\ 
	&\leq 	\Pr\left[ \left|\Cnt{F}{h,\vecbold{\alpha}} - \mu_{m} \right| \geq \frac{1}{2} \mu_{m} \right] 
	\leq 8 \cdot \left( \frac{4(n\mu_{m} + n^2)}{\mu_{m}^2} \right)^{n/2}\\
	&\leq 8 \cdot \left(\frac{4*4097}{4096*4096}\right)^{n/2} \leq 8 \times 31^{-n} 
\end{align*}

In the above derivation, we note that we have made use of Proposition 2.2 stated
at the end of Section~\ref{sec:prelim}.

Thus, $\Pr[\exists \vecbold{\alpha}, \Cnt{F}{h,\vecbold{\alpha}} < 1024n \text{ or }  \Cnt{F}{h,\vecbold{\alpha}} >  16384n] \leq 8 \cdot \left(\frac{2}{31}\right)^{n}<1$, for $n\geq 1$, which implies $\Pr[\forall \vecbold{\alpha}, \Cnt{F}{h,\vecbold{\alpha}} > 1024n \text{ and }  \Cnt{F}{h,\vecbold{\alpha}} \leq  16384n] >0$. 
Since the probability space is defined over random choice of $h$ and  therefore, there must exist an $h$ such that we  have $\exists h \forall \vecbold{\alpha}, \Cnt{F}{h,\vecbold{\alpha}} > 1024n \text{ and }  \Cnt{F}{h,\vecbold{\alpha}} \leq  16384n$. Furthermore, observe $\ell = 1024n$ and $u = 16384n$, we have that $\varphi^{\langle F, \ell, u \rangle}_{Cells}(m)$ is True. 

\end{proof}
}

\subsection{Auditing $\mathsf{EqualCellsCounter}(F)$}

Unlike for Stockmeyer's algorithm, for Algorithm~\ref{algo:smallcells}, a single call to a $\Sigma_2^P$ oracle is sufficient for verification and this call is intrinsically different from the earlier algorithm.

\begin{algorithm}
	\caption{EqualCellsAudit$(F,\mathsf{EqualCellsCounter}(F))$}
	\begin{algorithmic}[1]
          \State $\mathsf{ck}\gets 0$; 
          $m=\lfloor{\log (\frac{\Cest}{n})\rfloor}-10$; $\ell \gets 1024n$; $u\gets 16384n$;
                 \State  $\mathsf{ck}\gets \twoqbfcheck (\varphi_{Cells}^{\langle F, \ell, u\rangle}(m)[h \gets \hashcells])$; \label{line:smallcellscheck}
		\If{$\mathsf{ck}==1$}\label{line:equalcheck}                            
                \State \Return $\mathsf{Verified}$;
\EndIf
\end{algorithmic}
        \label{algo:small-cells-audit}
\end{algorithm}

That is, we substitute in the formula $\varphi^{\langle F,\ell,u\rangle}_{Cells}$ the value returned by Algorithm~\ref{algo:smallcells} and the hash functions $\hashcells$ passed as witness and check if the resulting oracle call gives True. Now, examining the audit complexity, we can see that the number of variables on which the query in Algorithm~\ref{algo:small-cells-audit} is bounded by $n(2n+2)$: $n$ variables for $\vecbold{\alpha}$ (since $m\leq n$), $n^2+n$ for $u+1$ many $\vecbold{y_i}$'s and another $n^2$ for $\ell$ many $\vecbold{z_i}$'s, since $u$ and $\ell$ can at most be $n$. Thus, we get%
\begin{corollary}
	The audit complexity of Algorithm~\ref{algo:smallcells} is $O(n^2)$.
\end{corollary}
\fullversion{
\begin{proof}
 The audit formula to be checked for validity is  $\varphi^{\langle F, \ell, u \rangle}_{Cells}(m)$ after $m$ is substituted by $\lfloor{\log (\frac{\Cest}{n})\rfloor}-10$ and after $h$ is substituted by $\hashcells$.  Thus, we need $m$ Boolean variables
 for $\boldsymbol{\alpha}$ and $(\ell + u + 1)\cdot n$ Boolean variables for
 $\boldsymbol{y_1} \ldots \boldsymbol{y_{u+1}}$, $\boldsymbol{z_1}, \ldots \boldsymbol{z_\ell}$.  Since $u$ and $\ell$ are in $\bigO\big(n\big)$, and since $m \le n$, we get
 the audit complexity as $\bigO(n^2)$.
\end{proof}
}

\section{Improving the Audit Complexity}
\label{sec:combined}
We now focus on further improving the audit complexity. We start by observing that the factor $n^2$ in the audit complexity of Algorithm~\ref{algo:smallcells} was due to the size of cell being $O(n)$ wherein every solution requires $n$ variables to represent, and therefore, contributing to the factor of $n^2$. To this end, we ask if we can have cells of constant size; achieving this goal has potential to reduce the audit complexity to $O(n)$.

However, as we will see, having cells of size 1 suffices only for an $n$-factor approximation of the count. To achieve constant factor approximation, we rely on the standard technique of making multiple copies (as used in Section~\ref{sec:stock}), which will finally achieve an audit complexity of $O(n \log n)$. 

Taking a step back, recall that Stockmeyer's approach too relied on cells of size 1; after all, $\varphi_{stock}^{F}(m) $ asserted that for every solution $\vecbold{z}$, there is a hash function $h$ that maps $\vecbold{z}$ to a cell where it is the only solution. But the higher audit complexity of Stockmeyer's algorithm was due to the need to certify the ``False'' answer. So, in order to achieve best of both worlds: i.e., have cells of size 1 and better audit complexity, we construct a new formula $\varphi_{holes}^{F}(m) $ that operates on $F$ and evaluates to True for sufficiently many values of $m$ for which $\varphi_{stock}^{F}(m)$ would evaluate to False, in the hope that certifying the "True" answer of $\varphi_{holes}^{F}(m)$ would be easier: it indeed turns out to be true, as we will see later. 

\begin{align*}
	\varphi_{holes}^{F}(m) :=  \exists &h_1, h_2, \ldots h_{m},h_{m+1} \\ 
	&\forall \,\, \vecbold{\alpha} \,\, \exists \vecbold{z}  \left(\bigvee_{i} F(\vecbold{z}) \wedge h_i(\vecbold{z}) = \vecbold{\alpha} \right)	
\end{align*}

Essentially, the above formula is True if and only if there are $m+1$ hash functions such that for
every cell $\vecbold{\alpha}$, there is a solution that is mapped to the cell by some hash function.

We now turn to the design of an approximate counter, called {\toolname}, with a much improved audit complexity. 
Given an input formula $F$, we will first construct $F'$ by making $\log n$ copies. We can design an approximate counting algorithm that queries $\varphi_{stock}^{F'}(m)$ as well as $\varphi_{holes}^{F'}(m)$ for all values of $m$ and generate a certificate that consists of (1) smallest value of $m$ for which $\varphi_{stock}^{F'}(m)$  evaluates to True, denoted $c_{high}$, and (2) the largest value of $m$ for which $\varphi_{holes}^{F'}(m)$ evaluates to True, denoted $c_{low}$. A surprising aspect of our algorithm design is that we are able to compute the estimate of $\Cest$ without $c_{low}$, i.e., purely based on the $c_{high}$, but the certificate returned by {\toolname} uses $c_{low}$ and associated set of hash functions. 

\begin{algorithm}[h]
	\caption{$\toolname(F)$}
	\label{algo:pigeons}
	
	\begin{algorithmic}[1]
		\State $ \clow\gets 0 ; \chigh \gets n' ; \quad  n' \gets n \cdot \log n$; $\Cest\gets 0$;
		\State $F' \gets \mathsf{MakeCopies}(F,\log n)$
		\State $ \mathsf{hashassgn_{stock}} \gets \emptyset; \quad  \mathsf{hashassgn_{holes}} \gets \emptyset$
		\For{$m=1$ to $n'$}
		\State $(\mathsf{assign},\mathsf{ret}) \gets \threeqbfcheck(\varphi^{F'}_{holes}(m))$ \label{line:low-qbf-check}
		\If{$\mathsf{ret}$ == 0}\label{line:low-check}
		\State $\clow = m-1$
		\State \textbf{break}
		\EndIf
		\State $ \mathsf{hashassgn_{holes}} \gets \mathsf{assign}$
		\EndFor 
		
		\For{$m'=1$ to $n'$}
		\State $(\mathsf{assign}, \mathsf{ret}) \gets \twoqbfcheck (\varphi^{F'}_{stock}(m'))	$ \label{line:high-qbf-check}
		\If{$\mathsf{ret}$==1}\label{line:high-check}
		\State $\chigh = m'$
		\State $\mathsf{hashassgn_{stock}} \gets \mathsf{assign}$ 
\State \textbf{break}
\EndIf
\EndFor
\State $\Cest\gets 2^{\frac{c_\mathsf{high}}{\log n}}$  
		\State \Return ($\Cest, \clow,\chigh,\mathsf{hashassgn_{stock}}, \mathsf{hashassgn_{holes}}$)
	\end{algorithmic}

\end{algorithm}

\begin{theorem}
	Algorithm~\ref{algo:pigeons} solves DAC and makes $O(n \log n)$ calls to $\Sigma_3^P$ oracle.
\end{theorem}
\begin{proof}
Observe that \toolname\ returns the same output as that of Algorithm~\ref{algo:stock}, therefore, the correctness follows from the correctness of Algorithm~\ref{algo:stock}. Furthermore, observe that  \toolname\ makes at most $n\log  n$ calls to $\Sigma_2^P$ oracle calls and at most $n\log n$ $\Sigma_3^P$ oracle calls.
\end{proof}

But even though the correctness of the algorithm followed easily from earlier results, we will now see that proving the witness, i.e., the audit requires more work. To this end, we will show a relationship between $c_{high}$ and $c_{low}$.

We first provide sufficient condition for $\varphi_{holes}^{F}(m)  = 1$. 

\begin{lemma}[Sufficient]\label{lm:holesexist}
	If $\sol{F} \geq 2^{m+3}$,  then $	\varphi_{holes}^{F}(m)  = 1$. 
\end{lemma}

\begin{proof}
	Let us fix a cell $\vecbold{\alpha}$ and $h_i \xleftarrow{R} H(n,m,2)$. For $\vecbold{y} \in \{0,1\}^n$, define the indicator variable $\gamma_{\vecbold{y},\vecbold{\alpha},i}$ such that 
$\gamma_{\vecbold{y},\vecbold{\alpha},i} = 1$ if  $h_{i}(\vecbold{y}) = \vecbold{\alpha}$ and $0$ otherwise.
	Let $\Cnt{F}{h_i,\vecbold{\alpha}} = \sum\limits_{y	\in \satisfying{F}} \gamma_{\vecbold{y},\vecbold{\alpha},i}$. 
	Thus $\mu_{m} = \sum\limits_{y	\in \satisfying{F}} \expect[\gamma_{\vecbold{y},\vecbold{\alpha},i}] = \frac{\sol{F}}{2^m}$.
	
	By assumption we have $\sol{F} \geq 2^{m+3}$, which implies that $\mu_{m} \geq 8$. Further, since $h_i$ is chosen from a 2-wise independent hash family, we infer that $\{\gamma_{\vecbold{y},\vecbold{\alpha},i}\}$ are $2$-wise independent, and can therefore conclude that $\sigma^2[\Cnt{F}{h_i,\vecbold{\alpha}}] \leq \mu_{m}$.
	
	Using Payley-Zygmund inequality, we have
	\begin{align*}
		\Pr[\Cnt{F}{h_i,\vecbold{\alpha}} < 1] &\leq \Pr[\Cnt{F}{h_i,\vecbold{\alpha}} <  \frac{1}{8} \cdot \mu_{m}] \\
		&
		\leq 1 - \left(1-\frac{1}{8}\right)^2\frac{\mu_{m}^2}{\sigma^2[\Cnt{F}{h_i,\vecbold{\alpha}}]+\mu_{m}^2}  \\
		&\leq 1 - \left(1-\frac{1}{8}\right)^2\frac{\mu_{m}}{1+\mu_{m}}   \\
	\end{align*}
		As $h_1, h_2, \ldots h_{m+1}$ are chosen independently, we have  
	\begin{align*}
		\Pr\left[\bigcap_{i=1}^{m+1} \{\Cnt{F}{h_i,\vecbold{\alpha}} <  1 \} \right] &=\left(\Pr\left[\{\Cnt{F}{h_i,\vecbold{\alpha}} <   1 \} \right]\right)^{m+1}  \\
		&\leq \frac{1}{2^{m+1}}
	\end{align*}  Finally, we apply union bound to bound 
	$\Pr\left[\exists \vecbold{\alpha}, \ \bigcap_{i=1}^{m+1}  \{\Cnt{F}{h_i,\vecbold{\alpha}} <   1 \} \right] \leq \frac{2^{m}}{2^{m+1}} \leq \frac{1}{2}$. Therefore, 
	$\Pr\left[\forall \vecbold{\alpha}, \ \bigcup_{i=1}^{m+1}  \{\Cnt{F}{h_i,\vecbold{\alpha}}  \ge  1 \} \right] \geq \frac{1}{2}$.
	
	Observe the probability space is defined over the random choice of $h_1, h_2, \ldots h_{m+1}$, and therefore, since the probability of randomly chosen $h$ ensuring $\forall \vecbold{\alpha}, \ \bigcup_{i=1}^{m+1}  \{\Cnt{F}{h_i,\vecbold{\alpha}}  \ge  1 \}$ is non-\vecbold{z}ero, therefore, there must be exist an $h_1, h_2, \ldots h_m$. 
	Formally,  when $\sol{F} > 2^{m+3}$, we have
	
	\begin{align*}
		\exists h_1, h_2, \ldots h_{m+1} \forall \,\, \vecbold{\alpha} \,\, \exists \vecbold{z}  \left(\bigvee_{i} F(\vecbold{z}) \wedge h_i(\vecbold{z}) = \vecbold{\alpha} \right)	
	\end{align*}
\end{proof}

From this we can infer the following fact.

\begin{lemma}
	$c_{high}-c_{low} \leq 7$.
\end{lemma}
\shortversion{
\begin{proof}
    Let $m^* = \lfloor  \log |\satisfying{F'}|  \rfloor $. Then, from Lemma ~\ref{lm:stockexist}, we have $c_{high} \leq m^*+3$. Further, from Lemma~\ref{lm:holesexist}, we obtain $c_{low} \geq m^*-4$. Thus, we have  	$c_{high}-c_{low} \leq 7$
\end{proof}
}
\fullversion{
\begin{proof}
    Let $m^* = \lfloor  \log |\satisfying{F'}|  \rfloor $.
    Let $m'=m^*+3$, then observe that we have $|\satisfying{F'}| \leq 2^{m'-2}$, therefore, $\varphi_{stock}^{F'}(m') = 1$ from Lemma~\ref{lm:stockexist}. Recall that $c_{high}$ is the smallest value of $m$ such that $\varphi_{stock}^{F'}(m) = 1$. Therefore, $c_{high} \leq m' = m^*+3$.  

Similarly, let $\hat{m} = m^*-4$, then observe that $|\satisfying{F'}| \geq 2^{\hat{m}+3}$. Therefore, as per Lemma~\ref{lm:holesexist}, we have $\varphi_{holes}^{F'}(\hat{m}) = 1$. Recall that $c_{low}$ is the largest value of $m$ for which $\varphi_{holes}^{F'}(m) = 1$. Therefore, $c_{low} \geq \hat{m} = m^*-4$. Thus, we have  	$c_{high}-c_{low} \leq 7$.
\end{proof}
}

\subsection{Audit complexity of Algorithm~\ref{algo:pigeons}}

Now, we focus on certifying Algorithm~\ref{algo:pigeons}. Observe that, $c_{high}$ is the smallest value of $m$ for which $\varphi_{stock}^{F'}(m)$  evaluates to True and $c_{low}$ is the largest value of $m$ for which $\varphi_{holes}^{F'}(m)$ evaluates to True. Thus, a straightforward approach would be to certify that  $c_{high}$ is indeed smallest value of $m$ for which $\varphi_{stock}^{F'}(m)$ but that is precisely what we did in Algorithm~\ref{algo:stock-audit}, which led to the audit complexity of $O(n^2 \log ^2 n)$. How do we overcome the barrier? 

The key surprising observation is that we do not need to certify that  $c_{high}$ is indeed the smallest value: instead we will just certify that  $\varphi_{stock}^{F'}(c_{high})$ evaluates to True and $\varphi_{holes}^{F'}(c_{low})$ evaluates to True, and then we check $c_{high} - c_{low} \leq 7$, and these three facts allow the auditor to certify the estimate returned by Algorithm~\ref{algo:pigeons}. Note that we were able to side step certifying $c_{high}$ being the smallest value because we had access to $c_{low}$, which was not the case when we were trying to certify Algorithm~\ref{algo:stock}.

\begin{algorithm}
	\caption{CountAuditor$(F',\toolname(F))$}
	  \label{algo:lonely-audit}
	\begin{algorithmic}[1]
\State $\mathsf{poscheck}\gets 0$; $\mathsf{negcheck}\gets 0$;
\State  $\mathsf{poscheck}\gets \conpcheck (\varphi_{stock}^{F'}(\chigh)[h_1\ldots,h_{\chigh} \gets \hashstock])$; \label{line:zconpcheck}
                \State  $\mathsf{negcheck}\gets \twoqbfcheck(\varphi_{holes}^{F'}(\clow)[h_1\ldots h_{\clow}\gets \mathsf{hashassign}_{holes}])$;\label{line:ztwoqbfcheck}   \label{line:zconncheck}
		\If{$(\mathsf{poscheck}\land \mathsf{negcheck}==1)\land (\chigh-\clow \leq 7)$}\label{line:zequalcheck}                            
               \Return $\mathsf{Verified}$;
\EndIf
\end{algorithmic}
      
\end{algorithm}

To show correctness of CountAuditor, we also need the following, which gives a lower bound on no. of solutions:
\begin{lemma}[Necessary]
	\label{lem:holes}
	If 	$\varphi_{holes}^{F'}(m) =1$, then $\sol{F'}\geq \frac{2^{m}}{m+1}$ 
\end{lemma}
\begin{proof}
	Let us construct a many-to-one mapping $g$ from $\{0,1\}^m$ to $\satisfying{F'}$ such that every edge  from $\vecbold{\alpha} \in \{0,1\}^m$ to $\vecbold{z} \in \satisfying{F'}$ is labeled with a non-empty set of hash functions $\{h_i\}$ such that $h_i(\vecbold{z}) = \vecbold{\alpha}$.  %
	Such a mapping must exist when $\varphi_{holes}^{F'}(m) =1$, since for every cell $\vecbold{\alpha}$, we have  $\vecbold{z} \in \satisfying{F'}$ such that $h_i(\vecbold{z})  = \vecbold{\alpha}$ for some $h_i$.  Furthermore, observe that there can be at most $m + 1$ different $\vecbold{\alpha}$ that map to same $\vecbold{z}$ as every hash function maps $\vecbold{z}$ to only one unique $\vecbold{\alpha}$ . Therefore, the range of $g$, which is a subset of $\satisfying{F'}$, must be at least $\frac{2^m}{m  + 1}$.
\end{proof}

Finally, combining all these, we obtain:

\begin{theorem}
\label{thm:finalaudit}
  The audit complexity of Algorithm~\ref{algo:pigeons} is $O(n \log n)$.
\end{theorem}
\shortversion{
\begin{proof}
	We will first show that Algorithm~\ref{algo:lonely-audit} certifies the estimate returned by Algorithm~\ref{algo:pigeons}, and then focus on measuring the audit complexity. From the True answers to $\mathsf{poscheck}$ and $\mathsf{negcheck}$, the auditor can conclude using Lemma~\ref{lem:holes} and Claim~\ref{cl:stockup} that
\begin{align}
  \frac{2^{\clow}}{\clow + 1} \leq \sol{F'}\leq \chigh \cdot 2^{\chigh}
  \label{eq:imp}
\end{align}

From the rhs of Eq~\ref{eq:imp}, and since $\chigh\leq n \log n \leq n^2$, we  have 
$  \frac{\sol{F'}}{n^2} \leq \frac{\sol{F'}}{\chigh}\leq  2^{\chigh}$

Now, from line~\ref{line:zequalcheck} of the audit, we have $\chigh-\clow\leq 7$. We combine this with lhs of Eq~\ref{eq:imp}, and observe that $\clow\leq n \log n$, and hence $\clow + 1 \leq n^2$ for $n\geq 12$, $n^2>128$, 

\begin{align*}
  2^{\chigh} \leq n^2 \cdot 2^{\clow} \leq n^2\cdot  (\clow + 1) \cdot  \sol{F'} \leq n^4\cdot \sol{F'}
\end{align*}

Noting that $\sol{F'}=\sol{F}^{\log n}$, we have
\begin{align*}
  \big(\frac{\sol{F}}{4}\big)^{\log n} \leq (2^{\frac{\chigh}{\log n}})^{\log n} \leq (16\cdot \sol{F})^{\log n}\\
  \frac{\sol{F}}{4} \leq \Cest \leq 16\cdot \sol{F}
  \end{align*}

which is a 16-factor approximation. As before, the calls in line~\ref{line:zconpcheck} and line~\ref{line:zconncheck} can be combined and answered using a single $\Sigma_2^P$ oracle query. Finally, the number of variables on which this query is made is $O(|\vecbold{\alpha}|+|\vecbold{z}|)$, thus of size $O(n')=O(n\log n)$.
}

\fullversion{
\begin{proof}
	We will first show that Algorithm~\ref{algo:lonely-audit} certifies the estimate returned by Algorithm~\ref{algo:pigeons}, and then focus on measuring the audit complexity. We start by noting that $\clow, \chigh$ are at most $m\leq n\log n$. From the True answers to $\mathsf{poscheck}$ and $\mathsf{negcheck}$, the auditor can conclude using Lemma~\ref{lem:holes} and Claim~\ref{cl:stockup} that
\begin{align}
  \frac{2^{\clow}}{\clow + 1} \leq \sol{F'}\leq \chigh \cdot 2^{\chigh}
  \label{eq:imp}
\end{align}

From the rhs of Eq~\ref{eq:imp}, and since $\chigh\leq n \log n \leq n^2$, we  have 
$  \frac{\sol{F'}}{n^2} \leq \frac{\sol{F'}}{\chigh}\leq  2^{\chigh}$

Now, from line~\ref{line:zequalcheck} of the audit, we have $\chigh-\clow\leq 7$. We combine this with lhs of Eq~\ref{eq:imp}, and observe that $\clow\leq n \log n$, and hence $\clow + 1 \leq n^2$ for $n\geq 12$, $n^2>128$, 

\begin{align*}
  2^{\chigh} \leq n^2 \cdot 2^{\clow} \leq n^2\cdot  (\clow + 1) \cdot  \sol{F'} \leq n^4\cdot \sol{F'}
\end{align*}

Noting that $\sol{F'}=\sol{F}^{\log n}$, we have
\begin{align*}
  \big(\frac{\sol{F}}{4}\big)^{\log n} &\leq (2^{\frac{\chigh}{\log n}})^{\log n} \leq (16\cdot \sol{F})^{\log n}\\
  \frac{\sol{F}}{4} &\leq \Cest \leq 16\cdot \sol{F}
  \end{align*}

which is a 16-factor approximation. Notice that to obtain this guarantee, we indeed used not only the witnessing hash functions $\mathsf{hashassign_{stock}}$ and $\mathsf{hashassign_{holes}}$, but also critically, the fact that $\chigh-\clow\leq 7$. As before, the calls in line~\ref{line:zconpcheck} and line~\ref{line:zconncheck} of Algorithm~\ref{algo:lonely-audit} can be combined and answered using a single $\Sigma_2^P$ oracle query.

Finally, the number of variables on which this query is made is bounded by the number of variables in $\varphi^{F'}_{holes}$, where the hash functions are given as a witness. Hence the remaining variables for the query are just $\vecbold{\alpha}$ and $\vecbold{z}$. Thus the audit complexity of the Algorithm is bounded above by $O(|\vecbold{\alpha}|+|\vecbold{z}|)$, thus of size $O(n')=O(n\log n)$.
\end{proof}
  }

\section{Discussion and Perspectives}
\label{sec:conclusion}
The design of algorithms has traditionally focused on optimizing
resources like time, space, communication or access to random bits
without necessarily considering how hard it is to independently
certify the correctness of results.  Yet, the need for certified or
auditable results is becoming increasingly important in a world where
untrusted components are unavoidable in the implementation of any
non-trivial software.

In this paper, we presented a principled
approach towards addressing this problem by showing that appropriate
measures of audit complexity can be treated as a first-class
computational resource to be optimized at the design stage of an
algorithm itself, rather than as an afterthought.  For deterministic
approximate counting, this approach yields completely new algorithms
that are sub-optimal with respect to time/space resources, and yet
permit significantly simpler certification in practice.  We believe
the same ``trade off'' applies to other hard combinatorial problems as
well, and one needs to navigate the algorithm design space carefully
to strike the right balance between audit complexity and
time/space/communication complexity.  We hope our work opens the door
for further such investigations.  As concrete steps for future work,
we wish to explore the audit complexity versus time/space complexity
tradeoff in probabilistic approximate counting.

\paragraph*{Acknowledgments}
This work was supported in part by National Research Foundation Singapore under its NRF Fellowship Programme[NRF-NRFFAI1-2019-0004 ] , Ministry of Education Singapore Tier 2 grant MOE-T2EP20121-0011, and Ministry of Education Singapore Tier 1 Grant [R-252-000-B59-114 ].

\end{document}